\UseRawInputEncoding
\RequirePackage{fix-cm}

\documentclass[smallextended]{svjour3}       
\smartqed  
\usepackage{graphicx}
\usepackage{xcolor}
\usepackage{wrapfig}
\usepackage{subfig}
\usepackage{rotating}
\usepackage{url}
\usepackage{makeidx}
\usepackage{multirow}
\usepackage{stmaryrd} 
\usepackage{setspace}
\usepackage{algpseudocode}
\usepackage{algorithmicx}
\usepackage{amssymb}
\usepackage[justification=centering]{caption}
\usepackage{multirow}
\usepackage[flushleft]{threeparttable}
\usepackage{scrextend}
\usepackage[english]{babel}
\usepackage{blindtext}
\usepackage{adjustbox}
\usepackage{amsmath}
\usepackage{hyphenat} 
\usepackage{float}
\usepackage{braket}
\usepackage{tabularx}
\begin{document}

\title{An Efficient Simulation of Quantum Secret Sharing
}


\author{Kartick Sutradhar         \and
        Hari Om 
}


\institute{Kartick Sutradhar \at
              Indian Institute of Technology (ISM) Dhanbad \\
              Mobile: +91-7602621359\\
              \email{kartick.sutradhar@gmail.com}           
           \and
           Hari Om \at
              Indian Institute of Technology (ISM) Dhanbad \\
              \email{hariom4india@gmail.com} 
}

\date{Received: date / Accepted: date}

\maketitle

\begin{abstract}
In quantum cryptography, quantum secret sharing $(QSS)$ is a fundamental primitive. $QSS$ can be used to create complex and secure multiparty quantum protocols. Existing $QSS$ protocols are either at the $(n, n)$ threshold $2$ level or at the $(t, n)$ threshold $d$ level with a trusted player, where $n$ denotes the number of players and $t$ denotes the threshold number of players. Here, we propose a secure $d$-level $QSS$ protocol for sharing a secret with efficient simulation. This protocol is more secure, flexible, and practical as compared to the existing $QSS$ protocols: $(n, n)$ threshold $2$-level and $(t,n)$ threshold $d$-level with a trusted player. Further, it does not disclose any information about the secret to players. Its security analysis shows that the intercept-resend, intercept, entangle-measure, forgery, collision and collusion attacks are not possible in this protocol. 

\end{abstract}
\section{Introduction}
A dealer shares a secret with $n$ players in secret sharing $(SS)$, and when the secret needs to be reconstructed, the threshold number of players can do so collaboratively. The quantum secret sharing \cite{hillery1999quantum,bao2009threshold,yang2013secret,Gang4,lu2018verifiable,lau2013quantum,Hao,mashhadi2016fairly,dehkordi2019proactive,mashhadi2017provably,mashhadi2016share,mashhadi2016analysis,mashhadi2020toward,mashhadi2020csa,mashhadi2017new,karimifard2016semiquantum,charoghchi2021three,mashhadi2020improvement,shi2010quantum,run2010efficient,shi2011multi,gyongyosi2019quantum,sutradhar2020efficient} is a  fundamental primitive protocol for sharing a secret in quantum cryptography, which may be considered as an extension of secret sharing. The $QSS$ protocol can be used to create complex multiparty quantum computing protocols that are secure. In the $(n,n)$ threshold $QSS$, a dealer shares a secret with $n$ players by dividing it into $n$ bits, known as shares, which are distributed among $n$ players, each of whom has only his share. The secret can be reconstructed by the $n$ players working together. Similarly, in the $(t,n)$ threshold $QSS$, a dealer shares a secret with $n$ players by dividing it into $n$ bits and distributing them to the $n$ players. The $t$ players will work together to solve the mystery. Because it protects the quantum threshold and secure quantum multiparty computation, the $QSS$ is commonly used in quantum threshold cryptography and secure quantum multiparty computation.\\
Here, we propose a secure $d$-level $QSS$ protocol for sharing a secret, where $t$ players can reconstruct the secret without a trusted player. In our protocol, each player knows only his share, even the reconstructor knows only his share. In this protocol, we use some basic operations i.e., protocol-I of Shi {\em et al.} \cite{shi2016secure}, $CNOT$ gate \cite{nielsen2002quantum}, secure communication \cite{Gang1,Gang3,shi2017quantum,shi2016efficient,sun2020toward,shi2018efficient,peng2018novel,zhang2018economic,luo2018novel,xu2017nearest,sutradhar2020hybrid,sutradhar2020generalized,sutradhar2021efficient}, entangle state \cite{Gang2,shi2016comment,shi2016quantum,dan2016efficient,shi2016data,shi2015quantum,shi2015comments,shi2013multi,shi2012multiparty,shi2012novel,shi2011efficient,run2011novel,shi2011multi,shi2011asymmetric}, Quantum Fourier Transform $(QFT)$ \cite{Nielsen2002} and Inverse Quantum Fourier Transform $(QFT^{-1})$ \cite{Nielsen2002}, to transform the particles. We use a quantum approach in classical secret sharing to combine the benefits of both classical and quantum secret sharing, preventing attacks such as Intercept-Resend (IR), Intercept, Entangle-Measure (EM), Forgery, Collision, and Collusion.

\section{Related Work}
There are numerous $QSS$ protocols for secret sharing in quantum cryptography \cite{Mashhadi2019,mashhadi2012novel,hillery1999quantum,bao2009threshold,mashhadi2012analysis,yang2013secret,Gang4,lu2018verifiable,lau2013quantum,Hao,dehkordi2008new,dehkordi2008efficient,mashhadi2015two,dehkordi2008verifiable,mashhadi2017secure,mashhadi2015computationally,mashhadi2013novel}. In 1999, Hillery {\em et al.} discussed the first $QSS$ protocol \cite{hillery1999quantum} based on the Greenberger-Home-Zeilinger $(GHZ)$ state. In 2009, Li {\em et al.} introduced a $QSS$ protocol \cite{bao2009threshold} of secure direct communication. This protocol is $(t, n)$ threshold scheme but $2$ level. In 2013, Yang {\em et al.} introduced a $QSS$ protocol \cite{yang2013secret} based on the $QFT$. This protocol is $d$-level $(t, n)$ threshold scheme but it is not secure because each player broadcasts the results of the measurement at the last step. Because the measurement results contain information about the secret, if an attacker intercepts the measurement results, he may expose the secret or execute an intercept-resend attack. In 2015, Qin {\em et al.} discussed a $QSS$ protocol \cite{qin2015t} based on the phase shift operation, which is $2$-level $(t, n)$ threshold scheme. The protocols \cite{bao2009threshold} and \cite{qin2015t} are not secure because the unitary operation transforms the private information of player $P_{e-1}$ and then the transformed information is transmitted to player $P_e$. So, the players $P_{e-1}$ and $P_{e+1}$ collaboratively can retrieve the private information of  player $P_e$. In 2017, Song {\em et al.} discussed a $(t,n)$ threshold $d$-level $QSS$ protocol \cite{song2017t} based on some basic operations, i.e., $d$-level $CNOT$ gate, $QFT$, generalized Pauli operator, and $QFT^{-1}$. In that protocol, Alice (dealer) selects $Bob_1$ as a trusted reconstructor from the set of participants $\mathbb{B} = \{Bob_1, Bob_2, \dots, Bob_n \}$ and then selects a hash function $SHA1$ \cite{eastlake2001us} to compute the hash value of the secret (which is to be shared) and sends this hash value to the trusted reconstructor $Bob_1$. Here, $Bob_1$ can perform collision attack to reveal the secret. So, the security of this protocol is dependent on the trusted reconstructor $Bob_1$. The main problem of Song {\em et al.'s} protocol is that the reconstructor $Bob_1$ cannot recover the original secret because $QFT^{-1}$ cannot be summed up over all the states \cite{CommentKao2018}. In other words, the reconstructor $Bob_1$ needs the secret information of other players to reconstruct the original secret. In 2018, Qin {\em et al.} \cite{Qin2018Multidimensional} discussed a $QSS$ protocol which can efficiently share a secret by using the $QFT$ and Pauli operator, but it is a $(n, n)$ threshold scheme. 

In our protocol, any $t$ players can reconstruct the secret without a trusted player and  each player knows only his share, nothing else. Furthermore, the reconstructor is unable to perform the collision attack because the secret's hash value is shared among the players.
\section{Preliminaries}
The $QFT$, $QFT-1$, Control-NOT $(CNOT)$ gate, and Shamir's Secret Sharing, which will be used in the proposed $QSS$ protocol, are all introduced here.
\subsection{Quantum Fourier Transform}
The $QFT$ \cite{Nielsen2002}, a unitary transform, is based on the quantum phenomenon and expansion of the standard discrete Fourier transform. For $s \in \{0, 1, \dots d-1 \}$, the $QFT$ of $d$-level quantum system is defined as follows:
\begin{equation}\label{equ1}
QFT: \ket{s}  \rightarrow \frac{1} {\sqrt{d}} \sum_{q=0}^{d-1} e^{2\pi i\frac{s}{d}q} \ket{q}.
\end{equation}
The $QFT^{-1}$ is defined by
\begin{equation}\label{equ2}
QFT^{-1}: \ket{q}  \rightarrow \frac{1} {\sqrt{d}} \sum_{s=0}^{d-1} e^{-2\pi i\frac{q}{d}s} \ket{s}.
\end{equation}
Further,
\begin{equation}\label{equ3}
\sum_{q=0}^{d-1} e^{2\pi i\frac{s}{d}q} = \begin{cases} 0~\text{if}~s \ne 0~mod~d\\ d~\text{if}~s = 0~mod~d \end{cases}
\end{equation}
So,
\begin{equation}\label{equ4}
\begin{split}
QFT^{-1}\Bigg( \frac{1} {\sqrt{d}} \sum_{q=0}^{d-1} e^{2\pi i\frac{s}{d}q} \ket{q} \Bigg) & = \frac{1} {\sqrt{d}} \sum_{q=0}^{d-1} e^{2\pi i\frac{s}{d}q} QFT^{-1} \ket{q} \\
& = \frac{1} {d} \sum_{q=0}^{d-1} \ket{s} + \frac{1} {d} \sum_{k=0\wedge k\ne s}^{d-1} 0.\ket{k} = \ket{s}
\end{split}
\end{equation}
That is,
\begin{equation}\label{equ5}
QFT^{-1}(QFT\ket{s})=\ket{s}.
\end{equation}
\subsection{Control-NOT $(CNOT)$ gate}
The $CNOT$ gate \cite{nielsen2002quantum} is a two-qubit gate, one is control qubit and other is target qubit. If the control bit of $CNOT$ gate is set to $\ket{0}$, then the $NOT$ gate would not be applied to the target bit. If the control bit of the $CNOT$ gate is set to $\ket{1}$, then the $NOT$ gate would be applied to the target bit.
\subsection{Shamir's Secret Sharing}
In the Shamir's secret sharing \cite{shamir1979share}, there are a dealer $\mathbb{D}$ and $n$ players $\mathcal{P} = \{P_1, P_2, \dots P_n\}$. The Shamir's secret sharing consists of two phases:
\subsubsection{Secret Sharing Phase}
In this phase, the dealer selects a polynomial $f(x)=S+a_1x+a_2x^2+\dots+a_{t-1}x^{t-1}$ of degree $(t-1)$, where $S$ is a secret and $a_1, a_2, \dots, a_{t-1} $ are coefficients of the polynomial $f(x)$. The dealer computes $n$ shares and distributes them among $n$ players, each player $P_i$ only knows $f(x_i)$, where $i=1, 2, \dots, n$. 
\subsubsection{Secret Reconstruction Phase}
Using $t$ shares of the secret and the Lagrange interpolation formula, $t$ players will jointly reconstruct the secret in this phase.
\begin{equation}\label{equ6}
f(x) = \sum_{r=1}^{t} f(x_r) \prod_{1 \le j \le t, j \neq r} \frac{x - x_j}{x_r - x_j}
\end{equation}
To calculate the polynomial at $x=0$,  Eq.(\ref{equ6}) can be simplified as
\begin{equation}\label{equ7}
\begin{split}
f(0) &= \sum_{r=1}^{t} f(x_r) \prod_{1 \le j \le t, j \neq r} \frac{x_j}{x_j - x_r}
\end{split}
\end{equation}
\section{Proposed Method}
We present a $d$-level $QSS$ protocol for sharing a secret that allows $t$ players to reconstruct the secret without the help of a trusted player. In comparison to the existing $QSS$ protocols, such as the $(n,n)$ threshold $2$-level and $(t,n)$ threshold $d$-level, which both require a trusted player, this protocol is more secure, versatile, and practical.Furthermore, no information about the secret is revealed to any of the players. There are two stages to the $QSS$ protocol: secret sharing and secret reconstruction.
\subsection{Secret Sharing Phase}
In this phase, the dealer $\mathbb{D}$ shares the secret among players $\mathcal{P} = \{P_1, P_2, \dots P_n\}$. Initially, the dealer $\mathbb{D}$ selects a prime $d$ such that $2 \le d \le 2n$ and sets a finite field $Z_d$. Then, the dealer $\mathbb{D}$ selects a polynomial $f(x)=S + a_1x + a_2x^2 + \dots + a_{t-1}x^{t-1}$ of degree of $(t-1)$, where $S$ is secret, $a_1, a_2, \dots, a_{t-1}$ are coefficients of polynomial $f(x) \in Z_d$ and the symbol $'+'$ is defined as addition modulo $d$. The dealer computes the classical shares $f(x i)$ and uses the BB84 protocol to encode these classical shares $f(x i)$ in a qubit string \cite{bennett1984update}. The qubit string of $f(x_i)$ is distributed among $n$ players, player $P_i$ only knows the share $f(x_i)$. In addition, the dealer $\mathbb{D}$ selects the $SHA1$ hash function to compute the hash value $\mathcal{H}(S)$ \cite{eastlake2001us} and shares it among $n$ players using the polynomial $g(x)=\mathcal{H}(S) + b_1x + b_2x^2 + \dots + b_{t-1}x^{t-1}$. Player $P_i$ only knows the share $g(x_i)$, where $i= 1, 2, \dots, n$.
\subsection{Secret Reconstruction Phase}
Suppose $\mathcal{Q}=\{P_1, P_2 \dots P_t\}$ is a qualified subset from all the qualified subsets, where the number of players in each qualified subset is $t$. The dealer $\mathbb{D}$ selects a player from the qualified subset $\mathcal{Q}=\{P_1, P_2 \dots P_t\}$ as a reconstructor. Here, the dealer $\mathbb{D}$ selects player $P_1$ from the qualified subset $\mathcal{Q}=\{P_1, P_2 \dots P_t\}$ as a reconstructor. The reconstructor $P_1$ only knows his share, nothing else. This reconstructor $P_1$ reconstructs the secret and hash value. The process of reconstruction is given as follows:\\
\textbf{Step 1:} Player $P_r$,   $r= 1, 2, \dots, t$, calculates the shadow $(s_r)$ of the share as follows.
\begin{equation}\label{equ8}
s_r  = f(x_r) \prod_{1\leq j\leq t, j\neq r} \frac {x_j} {x_j - x_r} \mod d
\end{equation}
\textbf{Step 2:} Player $P_1$ (reconstructor) makes basis state $\ket{s_1}_H$, where size of the basis state is $c$-qubit, $s_1$ is his private shadow of the share and $c = \lceil \log_{2}^{d} \rceil$. Then, player $P_1$ applies $QFT$ on the state $\ket{s_1}_H$ and the resultant state $\ket{\varphi_1}$ is calculated as follows:
\begin{equation}\label{equ9}
\begin{split}
\ket{\varphi_1} & = (QFT\ket{s_1}_H)\\
& =\frac{1} {\sqrt{d}} \sum_{k=0}^{d-1} e^{2\pi i\frac{s_1}{d}k} \ket{k}_H
 \end{split}
\end{equation}
\textbf{Step 3:} Player $P_1$ again makes ancillary state $\ket{0}_T$, where size of the ancillary state is $c$-qubit and $c = \lceil \log_{2}^{d} \rceil$, and then executes  $CNOT^{\otimes c}$ operations on the combined state $\ket{\varphi_1} \ket{0}_T$, where the first $c$-qubits is control qubit and second $c$-qubits is target qubit. After performing $CNOT^{\otimes c}$ operations, the state $\ket{\varphi_1}$ evolves as an entangled state $\ket{\varphi_2}$, where subscript $H$ or $T$ represents home state (non-transmitted state) or transmitted state.
\begin{equation}\label{equ10}
\begin{split}
\ket{\varphi_2} &= CNOT^{\otimes c} \ket{\varphi_1} \ket{0}_T\\
&=\frac{1} {\sqrt{d}} \sum_{k=0}^{d-1} e^{2\pi i\frac{s_1}{d}k} \ket{k}_H \ket{k}_T
 \end{split}
\end{equation}
\textbf{Step 4:} Player $P_1$ communicates with player $P_2$ using the authenticated quantum channel to send the ancillary state $\ket{k}_T$ (i.e., second $c$-qubits).\\
\\
\textbf{Step 5:} Player $P_2$ applies an oracle operator $C_k$ on $\ket{k}_T \ket{s_2}$, where $C_k$ is given by
\begin{equation}\label{equ11}
C_k: \ket{k}_T \ket{s_2} \rightarrow \ket{k}_T U^k \ket{s_2}
\end{equation}
with 
\begin{equation}\label{equ12}
U \ket{s_2} = e^{2\pi i\frac{s_2}{d}} \ket{s_2}
\end{equation}
where, $\ket{s_2}$ is an eigenvector of $U$ with eigenvalue $e^{2\pi i\frac{s_2}{d}}$. The combined quantum system of $P_1$ and $P_2$ is shown as follows.
\begin{equation}\label{equ13}
\begin{split}
 \ket{\varphi_3} &= C_k \frac{1} {\sqrt{d}} \sum_{k=0}^{d-1} e^{2\pi i\frac{s_1}{d}k} \ket{k}_H \ket{k}_T \ket{s_2}\\
 &= \frac{1} {\sqrt{d}} \sum_{k=0}^{d-1} e^{2\pi i\frac{s_1 + s_2}{d}k} \ket{k}_H \ket{k}_T \ket{s_2}
\end{split}
\end{equation}
\textbf{Step 6:} Player $P_2$ communicates with player $P_3$ through an authenticated quantum channel to send the ancillary state $\ket{k}_T$ and keeps $\ket{s_2}$ as secret. Player $P_3$ performs  $t-1$ times similar process as done by $P_2$. If $t$ players honestly perform the protocol, then the combined quantum state is obtained as shown below.
\begin{equation}\label{equ14}
 \ket{\varphi_4} = \frac{1} {\sqrt{d}} \sum_{k=0}^{d-1} e^{2\pi i\big(\frac{\sum_{r=1}^{t}s_r}{d}\big)k} \ket{k}_H \ket{k}_T \ket{s_2} \dots \ket{s_t}.
\end{equation}\\
\textbf{Step 7:} The ancillary state $\ket {k}_T$ is sent by $P_t$ back to $P_1$ through  an authenticated quantum channel. Player $P_1$ again performs $CNOT^{\otimes c}$ operation on his $2c$ qubits, where the first $c$-qubits is control qubit and second $c$-qubits is target qubit. The output state is shown as below.
\begin{equation}\label{equ14}
\begin{split}
\ket{\varphi_5} &= CNOT^{\otimes c} \ket{\varphi_4} =\frac{1} {\sqrt{d}} \sum_{k=0}^{d-1} e^{2\pi i\big(\frac{\sum_{r=1}^{t}s_r}{d}\big)k} \ket{k}_H \ket{0}_T\ket{s_2} \dots \ket{s_t}
 \end{split}
\end{equation}
\textbf{Step 8:} The second $c$-qubits (i.e., ancillary state $\ket{0}_T$) is measured by player $P_1$ in computational basis. If the output of the measurement is $\ket{0}$, then player $P_1$ continues the process; otherwise, he believes that the protocol executes with at least one corrupted player and ends the protocol.\\
\\
\textbf{Step 9:} Player $P_1$ applies $QFT^{-1}$ on the first $c$-qubits and measures the output to get the secret $f(0)'=\sum_{r=1}^{t}~s_r~mod~d$.\\
\\
\textbf{Step 10:} Finally,  $t$ players perform all the above nine steps again to get the hash value of the secret and player $P_1$ gets the hash value of the secret $g(0)'=\sum_{r=1}^{t}~h_r~mod~d$, where $h_r$ is the shadow of hash value shares. Player $P_1$ uses the hash function $SHA1$ to compute the hash value $\mathcal{H}(f(0)')$ and compares it with the hash value $g(0)'$. If $(\mathcal{H}(f(0)')=g(0)')$, then player $P_1$ realizes that all $t$ players have performed the reconstruction phase honestly; otherwise,  player $P_1$ believes that there is at least one corrupted player.
\section{Correctness Proof of $(t, n)$ threshold $d$-level $QSS$}
Here, we prove the correctness of the proposed $(t, n)$ threshold $d$-level $QSS$. We mainly focus on the correctness proof of Step $9$ of secret reconstruction phase.
\begin{lemma}
If $QFT^{-1}$ (as given in Equation~\ref{equ2}) is applied to the first $c$-qubits, then the measurement of the output is secret $(f(0)')$.
\end{lemma}
 \begin{proof}
Applying $QFT^{-1}$ to the first $c$-qubits provides the process of secret recovery as given below:\\
The original secret $f(0)'$ can be calculated using the Lagrange interpolation and Equation~\ref{equ7} as follows.
\begin{equation}\label{equ15}
\begin{split}
f(0)' & = f(x_1) \prod_{1\leq j\leq t, j\neq 1} \frac {x_j} {x_j - x_1} + \dots + f(x_t) \prod_{1\leq j\leq t, j\neq t} \frac {x_j} {x_j - x_t} \mod d\\
& = (s_1+\dots+s_t) \mod d\\
& = (\sum_{r=1}^{t} s_r \mod d)
 \end{split}
\end{equation}
Player $P_1$ applies $QFT^{-1}$ to the first $c$-qubits.
\begin{equation}\label{equ16}
\begin{split}
QFT^{-1} \Bigg( \frac{1} {\sqrt{d}} \sum_{k=0}^{d-1} e^{2\pi i\big(\frac{\sum_{r=1}^{t}s_r}{d}\big)k} \ket{k}_H \Bigg) &= \frac{1} {\sqrt{d}} \sum_{k=0}^{d-1} e^{2\pi i\big(\frac{\sum_{r=1}^{t}s_r}{d}\big)k} QFT^{-1} \ket{k}_H\\
&=\Bigg | \sum_{r=1}^{t} s_r~mod~d \Bigg >_H + \frac{1} {d} \sum_{l=0}^{d-1}  0 . \ket{l}_H\\
&=\Bigg | \sum_{r=1}^{t} s_r~mod~d \Bigg >_H = \ket{f(0)'}_H
 \end{split}
\end{equation}
Therefore, if this protocol honestly is executed by $t$ players, the reconstructor $P_1$ will get the original secret.
\end{proof}
\section{Simulation Results}
In this protocol, the initiator $P_1$ applies  $QFT$  on the $q$-qubit state and executes $CNOT$ gate. Then, the ancillary qubit is sent to player $P_2$, who applies the oracle operator on the ancillary qubit. Thereafter,  player $P_2$ sends the ancillary qubit to $P_3$ and player $P_3$ performs similar process. This process is performed $(t-1)$ times. After that, the ancillary qubit is sent back to player $P_1$ by $P_t$. The player $P_1$ performs $CNOT$ and $QFT^{-1}$ to get the multiplication. In the secure multiparty quantum  multiplication, the Hadamard gate is taken to be the $QFT$. After this, the initiator $P_1$ performs $QFT$ on the $c$-qubit state and executes  $CNOT$ gate. Then, $\ket{k}_T$ is sent to player $P_2$ who applies the oracle operator on $\ket{k}_T$. Thereafter, player $P_2$ sends $\ket{k}_T$ to $P_3$, who performs similar process. This process is performed $(t-1)$ times. After that, the ancillary state $\ket{k}_T$ is sent back to player $P_1$ by $P_t$. Player $P_1$ performs $CNOT$ and $QFT^{-1}$ to get the multiplication. We have executed this quantum protocol for $(t,n)$ threshold secure multiparty multiplication using the following number of players and qubits: 
\begin{itemize}
 \item In simulations $1-3$, we have considered three players with one qubit, three players with two qubits, and three players with three qubits, respectively, and got efficient result after taking $8192$ number of average shots.
 \item In simulations $4-6$, we have considered four players with one qubit, four players with two qubits, and four players with three qubits, respectively, and got efficient result after taking $8192$ number of average shots.
 \item In simulations $7-9$, we have considered fifteen players with one qubit, fifteen players with two qubits, and fifteen players with three qubits, respectively, and got efficient result after taking $8192$ number of average shots.
\end{itemize}
We got efficient results of multiplication after taking $8192$ number of average shots.
\section{Results and Discussion}
In this section, we discuss the security and performance analysis of the proposed $(t,n)$ threshold $QSS$ protocol based on some properties.
\subsection{Security Analysis}
Here, we analyze the outside (i.e., outside eavesdropper wants to steal the private information of all players) and participant (i.e., attack from one or more dishonest players) attacks. We discuss four types of outside attacks (i.e., Intercept-Resend $(IR)$, Intercept, Entangle-Measure $(EM)$ and Forgery) and two types of participant attacks (i.e., Collision and Collusion) \cite{cai2019cryptanalysis,ting2009participant,Wang2008,wang2011security,wang2013cryptanalysis,wang2017security} .
\subsubsection{Outside Attack}
In this type of attack, an outside eavesdropper wants to steal the private information of all players. We discuss the Intercept-Resend $(IR)$, Intercept, Entangle-Measure $(EM)$ and Forgery attacks as follows.
\paragraph{Intercept-Resend $(IR)$ Attack:}
In intercept-resend attack, a player measures the quantum state, which is sent by another player and replaces this state with his own state and then sends the replacement state to  other players. In our proposed protocol,  player $P_1$ sends the ancillary state $\ket{k}_T$ to  dishonest player $P_2$ through an authenticated quantum channel and player $P_2$ wants to eavesdrop $P_1$'s shadow of the share $s_1$. If the ancillary state measured by  dishonest player $P_2$ in computational basis ${\ket{0}, \ket{1}, \dots, \ket{d-1}}$. The dishonest player $P_2$ can succeed to get  $\ket{l}_T$ with the probability of $1/d$, but the output of the measurement $k$ is totally independent of $P_1's$ share $s_1$. Further,  player $P_2$ sends the state $\ket{k}_T$ to  player $P_3$. Unfortunately,   $k$ does not possess any partial information about   $P_1$'s shadow of the share $s_1$. The dishonest player $P_2$ cannot get any information from the intercepted state, and similarly the dishonest player $P_3$ cannot get any information from the transmitted state $\ket{k}_T$. So, the intercept-resend attack is infeasible.
\paragraph{Intercept Attack:}
In this attack, the dishonest player $P_2$ wants to eavesdrop $P_1$'s shadow of the share $s_1$. The dishonest player $P_2$ can measure the output of the unitary operator (transformed state) because,  based on $QFT$, player $P_2$ knows that  player $P_1$'s shadow of the share state $\ket{s_1}$ has evolved as the ancillary state $\ket{k}_T$. So,  $QFT^{-1}$ can be performed on the ancillary state $\ket{k}_T$ by dishonest player $P_2$ to reveal  $s_1$. If the ancillary state measured by dishonest player $P_2$ in computational basis ${\ket{0}, \ket{1}, \dots, \ket{d-1}}$, then $P_2$ can succeed to get $\ket{l}_T$ with the probability of $1/d$, but $P_2$ cannot get $P_1$'s shadow of the share, because the global information cannot be extracted from the limited number of qubits. The entangled systems cannot be disentangled by the limited number of qubits. So, the attacker cannot get any information about $P_1$'s shadow of the share.
\paragraph{Entangle-Measure $(EM)$ Attack:}
The dishonest player $P_2$ performs a more complicated entangle-measure attack. Player $P_2$ prepares an ancillary state $\ket{0}_{P_2}$ that gets entangled with the transmitted state $\ket{k}_T$ using the local unitary operations. Then, player $P_2$ measures the entangle state to get the partial information about player $P_1$'s shadow of the share. After successful completion of honesty test, it can easily be deduced that $\eta_k=1$. After performing $\bar{U}_{TP_2}$, 
 $P_2$ sends $\ket{k}_T$ back to $P_1$ and measures the ancillary system after execution of $CNOT^{\otimes c}$ operation by player $P_1$. If  player $P_2$ measures the ancillary state $\ket{\phi(k)}_{P_2}$, $P_2$ cannot get any information about $P_1$'s shadow of the share $s_1$ because of entanglement of $\ket{k}_H$ and $\ket{\phi(k)}_{P_2}$. So, this attack is also infeasible.
\paragraph{Forgery Attack:}
In forgery attack, the participants can execute the protocol with the fake shares. The proposed $QSS$ protocol can prevent the forgery attack, which is one of the important issues, where the participants can provide the fake shares. If any dishonest player performs the Pauli operator with the fake shadow, the original secret cannot be reconstructed correctly. In the proposed protocol,  player $P_1$ uses the hash function $SHA1$ to compute the hash value $H(f(0)')$ and compares it with the hash value $g(0)'$. If $(H(f(0)')=g(0)')$, then $P_1$ shares the secret with other $t-1$ players; otherwise, $P_1$ realizes that at least one player performs the reconstruction phase dishonestly and terminates the reconstruction phase. So, the forgery attack is not possible in our quantum ($t,n$) threshold $QSS$ protocol.
\subsubsection{Participant Attack}
This type of attack is performed by one or more dishonest players to reveal the secret information.
\paragraph{Collision (attack from one) Attack:}
In collision attack, the attacker performs an attack on the hash function, where the hash function produces the same hash value for two different inputs. Many existing $QSS$ protocols cannot prevent the collision attack. In \cite{song2017t}, Alice (dealer) selects $Bob_1$ as a trusted reconstructor from the set of participants $\mathbb{B} = \{Bob_1, Bob_2, \dots, Bob_n \}$ and then selects a hash function $SHA1$ to compute the hash value of the secret (which is to be shared). Then, Alice sends this hash value to the trusted reconstructor $Bob_1$. At this point, $Bob_1$ can perform collision attack to reveal the secret. So, the security of their protocol is dependent on the trusted reconstructor $Bob_1$.  In our protocol, the dealer $\mathbb{D}$ computes the hash value $H(S)$ using the $SHA1$ hash function and shares it among  $n$ players. Therefore, the reconstructor $P_1$ does not have any information about the hash value and he cannot perform the collision attack.
\paragraph{Collusion (attack from more than one dishonest players) Attack:}
In collusion attack, some players can collude together to get the shadow of the share of other player. In order to get the private information of $P_e$ ,  players $P_{e-1}$ and $P_{e+1}$ perform the protocol dishonestly. In our proposed protocol,  players $P_{e-1}$ and $P_{e+1}$ cannot perform the collusion attack because the unitary operation is performed by each participant with his private information. Moreover, this private information is not transmitted through a quantum channel.
\subsection{Performance Analysis}
We analyze and compare the performance of the proposed $QSS$ protocol with the existing $QSS$ protocols, i.e., Li {\em et al.'s} $QSS$ \cite{bao2009threshold}, Yang {\em et al.'s} $QSS$ \cite{yang2013secret}, Qin {\em et al.'s} $QSS$ \cite{qin2015t}, Song {\em et al.'s} $QSS$ \cite{song2017t}, and Qin {\em et al.'s} $QSS$ \cite{Qin2018Multidimensional} in terms of three parameters: universality, cost, and attack. Li {\em et al.'s} $QSS$  protocol \cite{bao2009threshold} is ($t,n$) threshold scheme, but it is not for $d$-level particle. Yang {\em et al.'s} $QSS$ protocol \cite{yang2013secret} is for $d$-level particle, but this protocol is ($n,n$) threshold scheme. Qin {\em et al.'s} $QSS$ protocol \cite{qin2015t} is ($t,n$) threshold scheme, but it is not for $d$-level particle. Song {\em et al.'s} $QSS$ protocol \cite{song2017t} is for $d$-level particle, and also it is ($t,n$) threshold scheme. Song {\em et al.'s} $QSS$ protocol \cite{song2017t} can prevent the $IR$, $EM$, forgery attacks, but it cannot prevent the collision attack. Qin {\em et al.'s} $QSS$ protocol \cite{Qin2018Multidimensional} is for $d$-level particle and $cn$ qubits, but this protocol is ($n,n$) threshold scheme. Here, we have compared the proposed $QSS$ protocol with six  existing protocols in terms of the communication cost and computation cost.The communication cost can be computed based on the transmitted particles, i.e., message particles and decoy particles. The computation cost can be computed based on five parameters: $QFT$, $U$ operation, $QFT^{-1}$, measure operation, and hash operation. The Li {\em et al.'s} $QSS$ protocol \cite{bao2009threshold} needs to perform $t(2t-1)$ number of $U$ operations, $t$ number of measure operations, needs to transmit $t(t+1)$ number of messages, and $z(t+1)$ number of decoy particles, where $z$ is the number of decoy particles. The Yang {\em et al.'s} protocol \cite{yang2013secret} needs to perform $n$ number of $QFT$, $n$ number of $U$ operations, $n$ number of measure operations, and needs to transmit $(n-1)$ number of message particles. The Qin {\em et al.'s} protocol \cite{qin2015t} needs to perform $t(t+1)$ number of $U$ operations, needs to transmit $t(t+1)$ number of messages, and $z(t+1)$ number of decoy particles. The Song {\em et al.'s} protocol \cite{song2017t} needs to perform $1$ number of $QFT$, $t$ number of $U$ operations, $1$ number of $QFT^{-1}$, $1$ number of measure operations, $2$ number of hash operations and needs to transmit $(t-1)$ number of message particles. The Qin {\em et al.'s} protocol \cite{Qin2018Multidimensional} needs to perform $1$ number of $QFT$, $t(t+1)+n$ number of $U$ operations, $1$ number of $QFT^{-1}$, $n$ number of measure operations, and needs to transmit $n$ number of message particles as well as $z(t+1)$ number of decoy particles. Our proposed protocol needs to perform $1$ number of $QFT$, $(t-1)$ number of $U$ operations, $1$ number of $QFT^{-1}$, $2$ number of hash operations, and needs to transmit $t$ number of decoy particles. So, the complexity of our proposed protocol is less as compared to the existing $QSS$ protocols.

\section{Conclusion}
In this paper, we discussed a secret-sharing protocol in which $t$ players can reconstruct the secret without the help of a trusted player. In comparison to existing $QSS$ protocols, our protocol is more secure, versatile, and practical. The reconstructor $P 1$ only knows his share and nothing else; even the secret's hash value is unknown to him. Since  the reconstructor $P_1$ only knows his share, it cannot perform the collision attack.


%
%


\section*{Ethical Statement}
This article does not contain any studies with human or animal subjects performed by the any of the authors. The manuscript has been prepared following the instructions provided in the Authors Guidelines of the journal.
\section*{Conflict of Interest}
The authors declare that they have no conflict of interest.

\end{document}